\documentclass[11pt,a4paper]{article}
\usepackage[a4paper,hmargin=2.8cm,vmargin=3cm]{geometry}
\usepackage[utf8x]{inputenc}
\usepackage{amsmath,amsthm,amsfonts,amssymb}

\usepackage[bookmarksnumbered=true]{hyperref}

\newtheorem{thm}{Theorem}[section]  

\newtheorem{prp}[thm]{Proposition}

\renewcommand{\theequation}{\thesection.\arabic{equation}}

\let\oldmarginpar\marginpar
\renewcommand\marginpar[1]{\-\oldmarginpar[\raggedleft\footnotesize #1]%
  {\raggedright\footnotesize #1}}

\newcommand{\fock}{\mathcal{F}}		
\newcommand{\di}{\textrm{d}}		
\newcommand{\Ncal}{\mathcal{N}}		
\newcommand{\scal}[2]{\big<#1,#2\big>} 
\newcommand{\Rbb}{\mathbb{R}}		
\newcommand{\Nbb}{\mathbb{N}}		

\newcommand{\tr}{\operatorname{tr}}

\newcommand{\tagg}[1]{ \stepcounter{equation} \tag{\theequation} \label{#1} } 

\newcommand{\eps}{\varepsilon}

\newcommand{\bR}{{\mathbb R}}

\newcommand{\be}{\begin{equation}}
\newcommand{\ee}{\end{equation}}

\newcommand{\cF}{{\cal F}}
\newcommand{\cH}{{\cal H}}

\newcommand{\cE}{{\cal E}}
\newcommand{\cN}{{\cal N}}

\newcommand{\bN}{{\mathbb N}}

\newcommand{\wt}{\widetilde}

\def \bea{\begin{eqnarray}}
\def \eea{\end{eqnarray}}

\title{Mean-Field Dynamics of Fermions with Relativistic Dispersion} 

\author{Niels Benedikter, Marcello Porta\thanks{Supported by ERC Grant MAQD 240518}{}\ \ and Benjamin Schlein\thanks{Partially supported by ERC Grant MAQD 240518}\\ \\ Institute of Applied Mathematics, University of Bonn\\ Endenicher Allee 60, 53115 Bonn, Germany}

\begin{document}
\maketitle

\begin{abstract}
We extend the derivation of the time-dependent Hartree-Fock equation recently obtained in \cite{BPS} to fermions with a relativistic dispersion law. The main new ingredient is the propagation of  semiclassical commutator bounds along the pseudo-relativistic Hartree-Fock evolution.
\end{abstract}
\section{Introduction}

We are interested in the dynamics of a system of $N$ fermions moving in three spatial dimensions with a relativistic dispersion law. In units with $\hbar = c = 1$, the evolution is governed by the Schr\"odinger equation
\begin{equation}\label{eq:schr1} i\partial_t \psi_{N,t} = \left[\sum_{j=1}^N \sqrt{-\Delta_{x_j} + m^2}  + \lambda\sum_{i < j}^N V(x_i -x_j) \right] \psi_{N,t} \end{equation}
for the wave function $\psi_{N,t} \in L_a^2 (\bR^{3N})$. In accordance with Pauli's principle $L^2_a (\bR^{3N})$ denotes the subspace of $L^2 (\bR^{3N})$ consisting of all functions which are antisymmetric with respect to any permutation of the $N$ particles. The function $V: \Rbb^3 \to \Rbb$ describes the two-body interaction among the particles.

\medskip

We are interested, in particular, in the mean-field limit, characterized by $N \gg 1$ and weak interaction $|\lambda| \ll 1$, so that $\lambda N^{2/3} = 1$ is fixed. For technical reasons, we also consider large masses $m$, keeping $m N^{-1/3} = m_0$ fixed in the limit. Introducing the semiclassical parameter $\eps = N^{-1/3}$, we can then rewrite (\ref{eq:schr1}) as
\begin{equation}\label{eq:schr2} i\eps \partial_t \psi_{N,t} = \left[ \sum_{j=1}^N \sqrt{-\eps^2 \Delta_{x_j} + m_0^2} + \frac{1}{N} \sum_{i<j}^N V (x_i -x_j) \right] \psi_{N,t} \end{equation}
{F}rom the physical point of view, it is important to understand the dynamics of initial data which can be easily prepared in labs. Hence, it makes sense to study the evolution of initial data close to the ground state of a Hamiltonian of the form 
\begin{equation}\label{eq:HNtrap} H_N^{\text{trapped}} = \sum_{j=1}^N \left[ \sqrt{-\eps^2 \Delta_{x_j} + m_0^2} + V_\text{ext} (x_j) \right] + \frac{1}{N} \sum_{i<j}^N V (x_i -x_j),\end{equation}
where $V_\text{ext}: \Rbb^3 \to \Rbb$ is an external potential, trapping the particles in a volume of order one. It is expected that the ground state of (\ref{eq:HNtrap}) can be approximated by the Slater determinant with one-particle reduced density $\omega_N$ minimizing the (relativistic) Hartree-Fock energy functional
\begin{equation}\label{eq:HFfunc}
\begin{split} \cE_{\text{HF}} (\omega_N) = \; &\tr \, \left[ \sqrt{-\eps^2 \Delta + m_0^2} + V_\text{ext} \right] \omega_N \\ &+ \frac{1}{2N} \int dx dy \, V(x-y) \left( \omega_N (x,x) \omega_N (y,y) - |\omega_N (x,y)|^2 \right) \end{split} \end{equation}
among all orthogonal projections $\omega_N$ on $L^2 (\bR^3)$ with $\tr \omega_N = N$ (recall that the reduced density of an $N$-particle Slater determinant is such an orthogonal projection). 

\medskip

In \cite{BPS}, we considered the evolution of $N$ non-relativistic fermions, governed by the 
Schr\"odinger equation
\begin{equation}\label{eq:schr-cl} i\eps \partial_t \psi_{N,t} = \left[ \sum_{j=1}^N -\frac{\eps^2 \Delta_{x_j}}{2m_0} + \frac{1}{N} \sum_{i<j}^N V(x_i -x_j)  \right] \psi_{N,t}.  \end{equation}
In particular, we were interested in the evolution of initial data close to Slater determinants 
minimizing a non-relativistic Hartree-Fock energy (similar to (\ref{eq:HFfunc}), but with a non-relativistic dispersion law). To this end, we argued that minimizers of the Hartree-Fock energy satisfy  semiclassical commutator estimates of the form
\begin{equation}\label{eq:semi} \tr \lvert[ x , \omega_{N}]\rvert \leq C N \eps, \quad \text{ and } \quad \tr \lvert[\eps \nabla, \omega_{N}]\rvert \leq C N \eps.\end{equation}
Motivated by this observation, we assumed initial data to be close to Slater determinants with reduced one-particle density satisfying (\ref{eq:semi}). For such initial data\footnote{In fact, instead of assuming $\tr\lvert [x,\omega_{N}]\rvert \leq C N \eps$, in \cite{BPS} we only imposed the weaker condition $\tr\lvert [e^{ip \cdot x}, \omega_{N}]\rvert \leq CN (1+|p|) \eps$, for all $p \in \bR$. In the present paper, however, we find it more convenient to work with the commutator $[x,\omega_N]$.}, we proved that for sufficiently regular interaction potential $V$ the many-body evolution can be approximated by the time-dependent non-relativistic Hartree-Fock equation
\begin{equation}\label{eq:HFcl} i \eps \partial_t \omega_{N,t} = \left[- \frac{\eps^2 \Delta}{2 m_0} + (V * \rho_t) - X_t , \omega_{N,t} \right]. \end{equation}
Here $\rho_t (x) = N^{-1} \omega_{N,t} (x,x)$ is the density of particles close to $x \in \bR^3$ and $X_t$ is the exchange operator, having the integral kernel $X_t (x,y) = N^{-1} V(x-y) \omega_{N,t} (x,y)$. We remark that, prior to \cite{BPS}, the convergence towards (\ref{eq:HFcl}) was proven in \cite{EESY} 
for analytic interactions and for short times (while convergence towards the Vlasov evolution for mean field dynamics of fermionic systems was established in \cite{NS,Sp}).

\medskip

In this paper we proceed analogously but for fermions with relativistic dispersion. Similarly as in the non-relativistic case, the arguments presented in \cite{BPS} and based on semiclassical analysis suggest that (approximate) minimizers of the Hartree-Fock energy (\ref{eq:HFfunc}) satisfy the commutator bounds (\ref{eq:semi}). For this reason, we will
consider the evolution (\ref{eq:schr2}) for initial data close to Slater determinants, with reduced density $\omega_N$ satisfying (\ref{eq:semi}). For such initial data, we will show in Theorem \ref{thm:main} below that the solution of the Schr\"odinger equation (\ref{eq:schr2}) stays close to a Slater determinant with one-particle reduced density evolving according to the relativistic Hartree-Fock equation
\begin{equation}\label{eq:HF} i\eps \partial_t \omega_{N,t} = \left[ \sqrt{-\eps^2 \Delta + m_0^2} + (V * \rho_t) - X_t , \omega_{N,t} \right], \end{equation}
where, like in (\ref{eq:HFcl}), $\rho_t (x) = N^{-1} \omega_{N,t} (x,x)$ and $X_t (x,y) = N^{-1} V(x-y) \omega_{N,t} (x,y)$. 

\medskip

For initial data minimizing the Hartree-Fock energy (\ref{eq:HFfunc}), the typical momentum of the particles is of order $\eps^{-1}$, meaning that the expectation of $\eps |\nabla|$ is typically of order one. Hence, for $m_0 \gg 1$, we can expand the relativistic dispersion as
\[ \sqrt{-\eps^2\Delta + m_0^2} = m_0 \, \sqrt{1- \frac{\eps^2 \Delta}{m_0^2}} \simeq m_0 \left(1- \frac{\eps^2\Delta}{2m_0^2} \right) = m_0 + \frac{-\eps^2 \Delta}{2m_0} \]
Since the constant $m_0$ only produces a harmless phase, this implies that in the limit of large $m_0$, one can approximate the solutions of the relativistic Schr\"odinger equation (\ref{eq:schr2}) and of the relativistic Hartree-Fock equation (\ref{eq:HF}) by the solutions of the corresponding non-relativistic equations (\ref{eq:schr-cl}) and, respectively, (\ref{eq:HFcl}). For fixed $m_0$ of order one, however, the relativistic dynamics cannot be compared with the classical evolutions. 

\medskip

If we start from (\ref{eq:schr1}) and consider the limit of large $N \gg 1$ and weak interaction $\lambda N^{2/3} = 1$ without scaling the mass $m$, we obtain a Schr\"odinger equation like (\ref{eq:schr2}), but with $m_0$ replaced by $\eps m$ (recall that $\eps = N^{-1/3}$). In the limit $N \gg 1$, this evolution can be compared with the massless relativistic Schr\"odinger equation
\begin{equation}\label{eq:schr3} i\eps \partial_t \psi_{N,t} = \left[ \sum_{j=1}^N \eps |\nabla_{x_j}| + \frac{1}{N} \sum_{i<j}^N V(x_i -x_j) \right] \psi_{N,t}. \end{equation}
In this case, we expect the dynamics of initial data close to Slater determinants satisfying the commutator estimates (\ref{eq:semi}) to be approximated by the Hartree-Fock equation
\begin{equation}\label{eq:HFml} i \eps \partial_t \omega_{N,t} = \Big[ \eps |\nabla| + (V* \rho_t) - X_t , \omega_{N,t} \Big]. \end{equation}
For technical reasons, we do not consider this case in the present work. Proving the convergence of (\ref{eq:schr3}) towards (\ref{eq:HFml}) remains an interesting open problem.

\medskip

\section{Main Result and Sketch of Proof}

To state our main theorem, we switch to a Fock space representation. We denote by 
\[ \cF = \bigoplus_{n \in \bN} L^2_a (\bR^{3n}, dx_1\dots dx_n) \]
the fermionic Fock space over $L^2 (\bR^3)$. For $f \in L^2 (\bR^3)$, we define creation and annihilation operators $a^* (f)$ and $a(f)$ satisfying canonical anticommutation relations 
\[ \left\{ a(f) , a^* (g) \right\} = \langle f,g \rangle, \quad \left\{ a(f) , a(g) \right\} = \left\{ a^* (f), a^* (g) \right\} = 0 \]
for all $f,g \in L^2 (\bR^3)$. We also use operator valued distributions $a_x^*, a_x$, $x \in \bR^3$. In terms of these distributions, we define the Hamilton operator
\begin{equation}\label{eq:hamfock} \cH_N = \int dx \, a_x^* \, \sqrt{-\eps^2 \Delta_x + m_0^2}  \;  a_x + \frac{1}{2N} \int dx dy \, V(x-y) a_x^* a_y^* a_y a_x. \end{equation}
We notice that $\cH_N$ commutes with the number of particles operator 
\[ \cN = \int dx \, a_x^* a_x.\]
When restricted to the $N$-particle sector, $\cH_N$ agrees with the Hamiltonian generating the evolution (\ref{eq:schr2}). 

\medskip

Let $\omega_N$ be an orthogonal projection on $L^2 (\bR^3)$, with $\tr \, \omega_N = N$. Then there are orthonormal functions $f_1, \dots , f_N \in L^2 (\bR^3)$ with $\omega_N = \sum_{j=1}^N |f_j \rangle \langle f_j|$. We complete $f_1, \dots , f_N$ to an orthonormal basis $( f_j )_{j\in \bN}$ of $L^2 (\bR^3)$. We define a unitary map $R_{\omega_N}$ on $\cF$. To this end, we denote by $\Omega = (1, 0 , \dots )$ the Fock vacuum and we set
\begin{equation}\label{eq:Romega2} R_{\omega_N} \Omega = a^\ast(f_1)\cdots a^\ast(f_N)\Omega, \end{equation}
a Slater determinant with reduced density $\omega_N$. Moreover we require that 
\begin{equation}\label{eq:Romega1} R^*_{\omega_N} a(f_i) R_{\omega_N} = \left\{ \begin{array}{ll} a (f_i) &\quad \text{if $i> N$} \\ a^* (f_i) &\quad \text{if $i \leq N$}. \end{array} \right. \end{equation}
The operator $R^*_{\omega_N}$ implements a fermionic Bogoliubov transformation on $\cF$. We consider the time evolution of initial data of the form $R_{\omega_N} \xi_N$, for a $\xi_N \in \cF$ with $\langle \xi_N, \cN \xi_N \rangle \leq C$ uniformly in $N$ (i.\,e.\ $R_{\omega_N} \xi_N$ is close to the $N$-particle Slater determinant $R_{\omega_N} \Omega$).

\medskip

We are now ready to state our main theorem.
\begin{thm}\label{thm:main}
Let $V \in L^1(\Rbb^3)$ with \begin{equation}\label{eq:ass} \int \lvert \widehat V(p)\rvert (1+\lvert p\rvert)^2 \di p < \infty.\end{equation}
Let $\omega_N$ be a sequence of orthogonal projections on $L^2 (\bR^3)$ with $\tr \, \omega_N = N$, satisfying the semiclassical commutator bounds (\ref{eq:semi}). Let $\xi_N$ be a sequence in $\cF$ with $\langle \xi_N , \cN \xi_N \rangle \leq C$ uniformly in $N$. We consider the time evolution 
\begin{equation}\label{eq:psiNt} \psi_{N,t} = e^{-i \cH_N t/ \eps} R_{\omega_N} \xi_N \end{equation}
generated by the Hamiltonian (\ref{eq:hamfock}), with $\eps = N^{-1/3}$ and with a fixed $m_0 > 0$. 
Here $R_{\omega_N}$ denotes the unitary implementor of a Bogoliubov transformation defined in (\ref{eq:Romega1}) and (\ref{eq:Romega2}). Let $\gamma_{N,t}^{(1)}$ be the one-particle reduced density associated with $\psi_{N,t}$. Then there exist constants $c,C > 0$ such that
\begin{equation}\label{eq:HSbd} \tr \left\lvert \gamma_{N,t}^{(1)} - \omega_{N,t} \right\rvert^2 \leq C \exp (c \exp ( c |t|)) \end{equation}
where $\omega_{N,t}$ is the solution of the time-dependent Hartree-Fock equation (\ref{eq:HF}) with initial data $\omega_{N,t=0} = \omega_N$. 
\end{thm}

\medskip

\noindent{\it Remarks:}
\begin{enumerate}
\item The bound (\ref{eq:HSbd}) should be compared with $\tr \, (\gamma^{(1)}_{N,t})^2$ and $\tr \, (\omega_{N,t})^2$, which are both of order $N$. The $N$-dependence in (\ref{eq:HSbd}) is optimal, since one can easily find a sequence $\xi_N \in \fock$ with $\scal{\xi_N}{\Ncal \xi_N} \leq C$ such that $\gamma_{N,0}^{(1)}-\omega_{N,0} = \mathcal{O}(1)$ (for example, just take $\xi_N = a^\ast(f_{N+1})\Omega$).
\item As in \cite{BPS}, under the additional assumptions that $\di\Gamma(\omega_N)\xi_N =0$ and $\scal{\xi_N}{\Ncal^2 \xi_N} \leq C$ for all $N \in \Nbb$, we find the trace norm estimate 
\begin{equation}\label{eq:trbd}\begin{split}
   \tr \lvert \gamma^{(1)}_{N,t} - \omega_{N,t} \rvert & \leq C N^{1/6}  \exp( c \exp(c \lvert t\rvert) ).
\end{split}\end{equation}
\item We can also control the convergence of higher order reduced densities.  If $\gamma^{(k)}_{N,t}$ denotes the $k$-particles reduced density associated with (\ref{eq:psiNt}), and if $\omega^{(k)}_{N,t}$ is the antisymmetric tensor product of $k$ copies of the solution $\omega_{N,t}$ of the Hartree-Fock equation (\ref{eq:HF}), we find, similarly to \cite[Theorem 2.2]{BPS},
\begin{equation}\label{eq:kbd} \tr \left\lvert \gamma_{N,t}^{(k)} - \omega_{N,t}^{(k)} \right\rvert^2 \leq C N^{k-1} \exp (c \exp (c |t|)). \end{equation}
This should be compared with $\tr \, (\gamma_{N,t}^{(k)})^2$ and $\tr (\omega_{N,t}^{(k)})^2$, which are of order $N^k$.
\item Just like in the non-relativistic model \cite[Appendix A]{BPS} the exchange term $[X_t,\omega_{N,t}]$ in the Hartree-Fock equation (\ref{eq:HF}) is of smaller order and can be neglected. The bounds (\ref{eq:HSbd}), (\ref{eq:trbd}), (\ref{eq:kbd}) remain true if we replace $\omega_{N,t}$ with the solution of the Hartree equation  
\begin{equation}
 \label{eq:hartree}
i\eps \partial_t \wt{\omega}_{N,t} = \left[\sqrt{-\eps^2 \Delta + m_0^2} + (V* \wt{\rho}_t), \wt{\omega}_{N,t} \right]
\end{equation}
with the density $\wt{\rho}_t (x) = N^{-1} \wt{\omega}_{N,t} (x,x)$. 
\item The relativistic Hartree-Fock equation (\ref{eq:HF}) and the relativistic Hartree equation (\ref{eq:hartree}) still depend on $N$ through the semiclassical parameter $\eps = N^{-1/3}$. As $N \to \infty$, the Hartree-Fock and the Hartree dynamics can be approximated by the relativistic Vlasov evolution. If $W_{N,t} (x,v)$ denotes the Wigner transform of the solution $\omega_{N,t}$ of (\ref{eq:HF}) (or, analogously, of the solution $\wt{\omega}_{N,t}$ of (\ref{eq:hartree})), we expect that in an appropriate sense $W_{N,t} \to W_{\infty ,t}$ as $N \to \infty$, where $W_{\infty,t}$ satisfies the relativistic Vlasov equation
\[ \partial_t W_{\infty,t} + \frac{v}{\sqrt{v^2 + m_0^2}} \cdot \nabla_x W_{\infty,t} - \nabla_v W_{\infty,t} \cdot \nabla \left( V * \rho_{\infty,t} \right) = 0 \, ,\]
where $\rho_{\infty,t} (x) = \int \di v\,W_{\infty,t} (x,v)$. In fact, the convergence of the relativistic Hartree evolution towards the relativistic Vlasov dynamics has been shown in \cite{AMS} for particles interacting through a Coulomb potential. In this case, however, a rigorous mathematical understanding of the relation with many-body quantum dynamics is still missing (because of the regularity assumption (\ref{eq:ass}), Theorem \ref{thm:main} does not cover the Coulomb interaction). In view of applications to the dynamics of gravitating fermionic stars (such as white dwarfs and neutron stars) and the related phenomenon of gravitational collapse studied in \cite{HS,HLLS}, this is an interesting and important open problem (at the level of the ground state energy, this problem has been solved in \cite{LY}). Notice that the corresponding questions for bosonic stars have been addressed in \cite{ES,MS}. 
 \end{enumerate}

\medskip

Next, we explain the strategy of the proof of Theorem \ref{thm:main}, which is based on the proof of Theorem 2.1 in \cite{BPS}. In fact, the main body of the proof can be taken over from \cite{BPS} 
without significant changes. There is, however, one important ingredient of the analysis of \cite{BPS} which requires non-trivial modifications, namely the propagation of the commutator bounds (\ref{eq:semi}) along the solution of the Hartree-Fock equation (\ref{eq:HF}). We will discuss this part of the proof of Theorem \ref{thm:main} separately in Section \ref{sec:semiclass}.  

\begin{proof}[Sketch of the proof of Theorem \ref{thm:main}]

We introduce the vector $\xi_{N,t} \in \cF$ describing the fluctuations around the Slater determinant with reduced density $\omega_{N,t}$ given by the solution of the Hartree-Fock equation (\ref{eq:HF}) by requiring that 
\[ \psi_{N,t} = e^{-i\cH_N t/\eps} R_{\omega_N} \xi_N =: R_{\omega_{N,t}} \xi_{N,t}.\]
This gives $\xi_{N,t} = U_N (t;0) \xi_N$, with the fluctuation dynamics
\[ U_N (t;s) = R_{\omega_{N,t}}^* e^{-i \cH_N t/\eps} R_{\omega_{N,s}}.\]
Notice that $U_N (t;s)$ is a two-parameter group of unitary transformations. The problem of proving that $\psi_{N,t}$ is close to the Slater determinant $R_{\omega_{N,t}} \Omega$ reduces to showing that the expectation of the number of particles in $\xi_{N,t}$ stays of order one, i.\,e.\ small compared to the $N$ particles in the Slater determinant. In fact, it is easy to check (see \cite[Section 4]{BPS}) the bound for the Hilbert-Schmidt norm   
\begin{equation}\label{eq:ncalest} \| \gamma_{N,t}^{(1)} - \omega_{N,t} \|_{\text{HS}} \leq C \langle \xi_{N,t}, \cN \xi_{N,t} \rangle = C \langle \xi_N , U_N^* (t;0) \cN U_N (t;0) \xi_N \rangle \, .  \end{equation}

\medskip

To bound the growth of the expectation of the number of particles with respect to the fluctuation dynamics $U_N (t;s)$ we use Gronwall's Lemma. Differentiating the expectation on the r.h.s. of (\ref{eq:ncalest}) with respect to time gives (see \cite[Proof of Prop.\ 3.3]{BPS}) 
\[ \begin{split} i\eps \frac{\di}{\di t} \langle \xi_N, U^*_N (t;0)& \cN U_N (t;0) \xi_N \rangle \\ &= \langle \xi_N, U_N^* (t;0) R_{\omega_{N,t}}^* \Big( d\Gamma (i\eps \partial_t \omega_{N,t}) - [ \cH_N , d\Gamma (\omega_{N,t})] \Big) R_{\omega_{N,t}} U_N (t;0) \xi_N \rangle \end{split} \]
where $d\Gamma (J)$ is the second quantization of the one-particle operator $J$, its action on the $n$-particle sector being given by $d\Gamma (J)|_{L^2_a(\Rbb^{3n})} = \sum_{i=1}^n J^{(i)}$, where $J^{(i)}$ denotes the operator acting as $J$ on the $i$-th particle and as the identity on the other $(n-1)$ particles. There are important cancellations between the two terms in the parenthesis. In particular, since
\[\begin{split}  \left[ \int \di x\,a_x^* \sqrt{-\eps^2\Delta + m_0^2} \, a_x \, , \, d\Gamma (\omega_{N,t}) \right] &= \left[ d\Gamma \left( \sqrt{-\eps^2 \Delta + m_0^2} \right), d\Gamma (\omega_{N,t}) \right] \\ &= d\Gamma \left( \left[ \sqrt{-\eps^2 \Delta + m_0^2} \, , \, \omega_{N,t} \right] \right) \end{split} \]
the contribution of the kinetic energy cancels exactly. The remaining terms are then identical to those found in the non-relativistic case. Hence, analogously to \cite[Prop.\ 3.3]{BPS}, we conclude that
\begin{equation}\label{eq:gron1} \begin{split} &i \eps \frac{d}{dt} \langle \xi_N, U_N^* (t;0) \cN U_N (t;0) \xi_N \rangle \\ &= - \frac{4i}{N} \text{Im } \int \di x \di y\, V(x-y) \big\{ \langle U_N (t;0) \xi_N , a^* (u_{t,x}) a (u_{t,x}) a (v_{t,y}) a (u_{t,y})  U_N (t;0) \xi_N \rangle  \\ &\hspace{4.7cm} + \langle U_N (t;0) \xi_N , a^* (v_{t,x}) a (v_{t,x}) a (v_{t,y}) a (u_{t,y}) U_N (t;0) \xi_N \rangle \\  &\hspace{4.7cm}  + \langle U_N (t;0) \xi_N , a (v_{t,x}) a (u_{t,x}) a (v_{t,y}) a (u_{t,y}) U_N (t;0) \xi_N \rangle \big\} \end{split} \end{equation}
where the functions $u_{t,x}$ and $v_{t,x}$ are defined by \[ R^*_{\omega_{N,t}} a_x R_{\omega_{N,t}} = a (u_{t,x}) +a^* (v_{t,x}).\] It is easy to express $u_{t,x}$ (which is actually a distribution) and $v_{t,x}$ (a $L^2$-function) in terms of $\omega_{N,t}$; see, for example, \cite[Eq.\ (2.27)]{BPS} (but notice that here we have replaced $\overline{v}$ with $v$). Notice that in \cite[Prop.\ 3.3]{BPS}, we also considered the expectation of higher moments of $\cN$. This can be done in the relativistic setting as well, and is needed to prove the trace-norm bound (\ref{eq:trbd}).

\medskip

Proceeding as in the proof of \cite[Lemma 3.5]{BPS}, we can bound the terms on the r.\,h.\,s.\ of (\ref{eq:gron1}) to show that
\begin{equation}\label{eq:gron2}\begin{split}  
\Big| i\eps \frac{d}{dt} \, \langle \xi_N, U^*_N &(t;0) (\cN+1) U_N (t;0) \xi_N \rangle \Big| \\ &\leq C N^{-1} \sup_{p \in \bR^3} \frac{\tr \lvert[e^{ip \cdot x}, \omega_{N,t}]\rvert}{1+|p|} \, \langle \xi_N, U^*_N (t;0) (\cN+1) U_N (t;0) \xi_N \rangle.
\end{split} \end{equation}
Using the integral representation
\[ [e^{ip\cdot x} , \omega_{N,t} ] = \int_0^1 ds \, e^{is p \cdot x} \, [ip \cdot x , \omega_{N,t}] \, e^{i(1-s) p \cdot x} \]
we conclude that
\begin{equation}\label{eq:exp-bd} \sup_{p \in \bR^3} \frac{\tr\lvert[e^{ip\cdot x} , \omega_{N,t}]\rvert}{1+\lvert p\rvert} \leq \tr\lvert[x, \omega_{N,t}]\rvert. \end{equation}
Hence, (\ref{eq:gron2}) implies the bound (\ref{eq:HSbd}) in Theorem \ref{thm:main}, if we can show that there exist constants $C,c > 0$ with 
\begin{equation}\label{eq:needed} \tr\, |[x,\omega_{N,t}]| \leq C N \eps \exp (c|t|) \end{equation}
for all $t \in \bR$. We show (\ref{eq:needed}) in Proposition \ref{prop:semi} below.
\end{proof}

\section{Propagation of the Semiclassical Structure}
\label{sec:semiclass}

The goal of this section is to show the estimate (\ref{eq:needed}), which is needed in the proof of Theorem \ref{thm:main}. To this end, we use the assumption (\ref{eq:semi}) on the initial data, and 
we propagate the commutator estimates along the Hartree-Fock evolution. This is the genuinely 
new part of the present paper, where the ideas of \cite{BPS} need to be adapted to the relativistic dispersion of the particles. 

\begin{prp}\label{prop:semi} Let $V \in L^1(\Rbb^3)$ with 
\begin{equation}\label{eq:ass2} \int \lvert \widehat V(p)\rvert (1+\lvert p\rvert)^{2} \di p < \infty.\end{equation}
Let $\omega_N$ be a trace-class operator on $L^2(\Rbb^3)$ with $0 \leq \omega_N \leq 1$ and $\tr \omega_N = N$, satisfying the commutator estimates (\ref{eq:semi}). Denote by $\omega_{N,t}$ the solution of the Hartree-Fock equation (\ref{eq:HF}) (with $\eps = N^{-1/3}$) with initial data $\omega_{N,0}= \omega_N$. Then there exist constants $C,c > 0$ such that
\[ \label{eq:ass3} \begin{split}
  \tr \lvert [ x, \omega_{N,t}] \rvert & \leq C N \eps \exp (c |t|) \quad \text{and}\\
 \tr \lvert [ \eps \nabla ,  \omega_{N,t} ]\rvert & \leq C N \eps \exp (c|t|)
 \end{split}\]
for all $t \in \bR$. 
\end{prp}
\begin{proof}
We define the Hartree-Fock Hamiltonian
\[h(t) := \sqrt{-\eps^2 \Delta + m_0^2}+(V\ast\rho_t) -X_t\]
where $\rho_t(x) = N^{-1} \omega_{N,t}(x,x)$ and $X_t$ is the exchange operator defined by the integral kernel $X_t (x,y) = N^{-1} V(x-y) \omega_{N,t} (x,y)$ (note that $\rho_t$ and $X_t$ depend on the solution $\omega_{N,t}$ of the Hartree-Fock equation (\ref{eq:HF})). Then $\omega_{N,t}$ 
satisfies the equation
\begin{equation}
 \label{eq:commutatorHF} i\eps \partial_t \omega_{N,t} = [ h(t),\omega_{N,t}].
\end{equation}
Using the Jacobi identity we obtain
\begin{equation}
 \label{eq:timeder}
\begin{split}
i \eps \partial_t [ x, \omega_{N,t}] & = [ x , [h(t),\omega_{N,t}]] \\& = [h(t),[ x,\omega_{N,t}]] + \left[ \omega_{N,t}, \left[ \sqrt{-\eps^2\Delta +m_0^2} \, , \, x \right] \right] - [\omega_{N,t},[X_t, x]].
\end{split}
\end{equation}
We can eliminate the first term on the r.h.s. of the last equation by conjugating $[x,\omega_{N,t}]$ with the two-parameter group $W(t;s)$ generated by the self-adjoint operators $h(t)$, satisfying 
\begin{equation}\label{eq:W} i \eps \partial_t W (t,s) = h(t) W (t,s) \quad \text{with} \quad W (s,s) = 1 \quad \text{for all $s \in \bR$}.\end{equation}
In fact, we have
\[ \begin{split}  i\eps \partial_t \, W^* (t;0) &[x, \omega_{N,t}] W(t;0) \\ & = W^* (t;0) \Big( \Big[ \omega_{N,t}, \Big[ \sqrt{-\eps^2\Delta +m_0^2} \, , \, x \Big] \Big] - [\omega_{N,t},[X_t, x]] \Big) W(t;0) \end{split} \]
and therefore
\begin{equation*}
\begin{split}
 W^*( &t,0)  [ x,\omega_{N,t}] W(t,0)\\
& = [x, \omega_{N,0}] + \frac{1}{i\eps} \int_0^t \frac{\di}{\di s} \big(W^*(s,0) [x, \omega_{N,s}] W(s,0) \big) \di s \\
& = [ x, \omega_{N,0}] + \frac{1}{i\eps} \int_0^t W^*(s,0) \Big( \Big[ \omega_{N,t}, \Big[ \sqrt{-\eps^2\Delta +m_0^2}, x \Big] \Big] - [\omega_{N,t},[X_t, x]] \Big) W(s,0) \, \di s.
\end{split}
\end{equation*}
This implies that
\begin{align*}
 \tr \lvert [ x,\omega_{N,t}]\rvert \leq \tr  \lvert [ x,\omega_{N,0}]\rvert & + \frac{1}{\eps} \int_0^t \di s\, \tr \Big\lvert \Big[\omega_{N,s}, \Big[\sqrt{-\eps^2 \Delta + m_0^2}, x \Big] \Big] \Big\rvert \tagg{eq:l1}\\
& + \frac{1}{\eps} \int_0^t \di s\, \tr \lvert [\omega_{N,s},[X_s, x]] \rvert \tagg{eq:l2}. 
\end{align*}
To control the term \eqref{eq:l2} we observe that
\begin{equation}\label{eq:exX} X_s = \frac{1}{N} \int dq \, \widehat{V} (q) \, e^{iq \cdot x} \omega_{N,s} e^{-iq \cdot x}, \end{equation}
where $x$ denotes the operator of multiplication by $x$. Since $\| \omega_{N,s} \| \leq 1$ (because of the assumption $0 \leq \omega_{N,s} \leq 1$, in accordance with Pauli's principle), we find
\begin{equation}\label{eq:328}\begin{split}
\tr \lvert[\omega_{N,s},[X_s, x]] \rvert & \leq  \frac{1}{N} \int \di q\, \lvert \widehat V(q)\rvert\,\tr \lvert[\omega_{N,s},[e^{iq\cdot x}\omega_{N,s} e^{-iq\cdot  x}, x]]\rvert \\
& \leq \frac{2}{N} \int \di q\, \lvert \widehat V(q)\rvert\,\tr \lvert[e^{iq\cdot  x}\omega_{N,s} e^{-iq\cdot x}, x]\rvert \\ &=  \frac{2}{N} \int \di q\, \lvert \widehat V(q)\rvert\,\tr \lvert e^{iq\cdot  x}[\omega_{N,s}, x]e^{-iq\cdot x}\rvert  \leq \frac{2 \| \widehat{V} \|_1}{N}  \tr \lvert [\omega_{N,s}, x]\rvert. 
\end{split}\end{equation}
To control \eqref{eq:l1} we notice that 
\[ \Big[\sqrt{-\eps^2 \Delta + m_0^2} \, , \, x \Big] = -\eps \frac{\eps \nabla}{\sqrt{-\eps^2 \Delta + m^2}}.\]
Hence
\[  \Big[ \omega_{N,s}, \Big[ \sqrt{-\eps^2\Delta+m_0^2}, x \Big] \Big] = -\eps [\omega_{N,s}, \eps \nabla] \frac{1}{\sqrt{-\eps^2 \Delta + m_0^2}} - \eps^2 \nabla \, \left[ \omega_{N,s}, \frac{1}{\sqrt{-\eps^2 \Delta + m_0^2}} \right] \]
and thus
\begin{equation} 
\label{eq:comm1}
\begin{split} \tr  \Big\lvert \Big[\omega_{N,s}, \Big[ \sqrt{-\eps^2\Delta+m_0^2}, x \Big] \Big] \Big\rvert \leq \; & \eps m_0^{-1} \tr \lvert [\eps\nabla, \omega_{N,s}]\rvert + \eps \tr \left\lvert \, \eps \nabla \left[ \omega_{N,s}, \frac{1}{\sqrt{-\eps^2 \Delta+ m_0^2}} \right] \right\rvert .\end{split}\end{equation}
Here we used the estimate $\| (-\eps^2 \Delta + m_0^2)^{-1/2} \| \leq m_0^{-1}$. To bound the second term on the r.\,h.\,s.\ we will use the integral representation 
\begin{equation}\label{eq:sqrtid}\frac{1}{\sqrt{A}} =  \frac{1}{\pi} \int_0^\infty \frac{\di \lambda}{\sqrt{\lambda}} (A+\lambda)^{-1}\end{equation}
and the identity
\[[(A+\lambda)^{-1}, B] = (A+\lambda)^{-1}[B,A](A+\lambda)^{-1} \]
for $A>0,B$ self-adjoint operators. 
Now consider the $j$-th component ($j \in \{ 1,2,3 \}$) of the operator whose trace norm we have to estimate:
\begin{align*}
\tr \Big\lvert \, \eps \partial_j &\left[ \omega_{N,s},\frac{1}{\sqrt{-\eps^2\Delta+m_0^2}} \right] \Big\rvert \\
& \leq \frac{1}{\pi}  \int_0^\infty \frac{\di \lambda}{\sqrt{\lambda}} \tr \left\lvert \eps\partial_j \frac{1}{{-\eps^2 \Delta+m_0^2+\lambda}} [\omega_{N,s}, \eps^2\Delta] \frac{1}{{-\eps^2\Delta+m_0^2+\lambda}} \right\rvert \\
& \leq  \frac{1}{\pi} \sum_{k=1}^3 \int_0^\infty \frac{\di \lambda}{\sqrt{\lambda}} \left\| \frac{-\eps^2 \partial_j \partial_k}{-\eps^2\Delta+m_0^2 +\lambda} \right\| \,  \tr\lvert [\omega_{N,s},\eps\partial_k] \rvert \left\| \frac{1}{-\eps^2\Delta+m_0^2+\lambda} \right\| \\
&\quad + \frac{1}{\pi} \sum_{k=1}^3 \int_0^\infty \frac{\di \lambda}{\sqrt{\lambda}} \left\| \frac{-i\eps\partial_j}{(-i\eps\nabla)^2+m_0^2+\lambda} \right\| \,  \tr\lvert [\omega_{N,s},\eps\partial_k] \rvert \, \left\| \frac{-i\eps\partial_k}{(-i\eps\nabla)^2+m_0^2+\lambda} \right\|. 
\end{align*}
Using the bounds $\| (-\eps^2 \Delta + m_0^2 + \lambda)^{-1} \| \leq (m_0^2 + \lambda)^{-1}$, 
\[ \begin{split} 
\left\| \frac{-i \eps \partial_k}{ -\eps^2 \Delta + m_0^2 + \lambda} \right\| \leq \frac{1}{\sqrt{m_0^2 + \lambda}} \quad \text{and } \quad
 \left\| \frac{- \eps^2 \partial_k \partial_j}{ -\eps^2 \Delta + m_0^2 + \lambda} \right\| &\leq 1, \end{split} \] 
all of which can be easily proved in Fourier space, we conclude that
\[ \begin{split} \tr \left\lvert \, \eps \partial_j \left[ \omega_{N,s},\frac{1}{\sqrt{-\eps^2\Delta+m_0^2}} \right] \right\rvert 
&\leq C \, \tr \, | [ \eps \nabla, \omega_{N,s}] | \, \int_0^\infty \frac{d\lambda}{\sqrt{\lambda}} \frac{1}{\lambda+m_0^2} \\ &\leq C m_0^{-1}  \tr \, | [ \eps \nabla, \omega_{N,s}] |.  \end{split} \]
Inserting this estimate in (\ref{eq:comm1}), we obtain
\[ \tr  \Big\lvert \Big[\omega_{N,s}, \Big[ \sqrt{-\eps^2\Delta+m_0^2}, x \Big] \Big] \Big\rvert \leq \; C \eps m_0^{-1} \tr \lvert [\eps\nabla, \omega_{N,s}]\rvert. \]
Plugging this bound and (\ref{eq:328}) into (\ref{eq:l1}) and (\ref{eq:l2}), we arrive at 
\begin{equation}
\label{eq:deriv1}
 \tr\lvert [ x,\omega_{N,t}]\rvert \leq \tr\lvert [ x,\omega_{N,0}]\rvert + Cm_0^{-1} \int_0^t \di s \, \tr \lvert [\eps\nabla, \omega_{N,s}] \rvert + C N^{-2/3} \int_{0}^t \di s \tr\lvert[x, \omega_{N,s}] \rvert.
\end{equation}

Next, we bound the growth of the commutator $[\eps \nabla, \omega_{N,t}]$. Since the kinetic energy commutes with the observable $\eps \nabla$, we can proceed here as in the non-relativistic case considered in \cite{BPS}. For completeness, we reproduce the short argument. Differentiating w.\,r.\,t.\ time and applying Jacobi identity, we find
\[ \begin{split}  i\eps \frac{d}{dt} \, [\eps \nabla , \omega_{N,t} ] =& [\eps \nabla, [h (t) , \omega_{N,t} ]] \\
=& [h (t) ,[\eps \nabla, \omega_{N,t}]] + [\omega_{N,t}, [ h (t), \hbar \nabla]] \\ =& [h (t) , [\eps \nabla, \omega_{N,t}]] + [\omega_{N,t}, [ V*\rho_t , \eps \nabla]] - [ \omega_{N,t}, [X_t, \eps \nabla]] \, .
\end{split} \]
As before, the first term on the r.\,h.\,s.\ can be eliminated by conjugation with the unitary maps $W(t;0)$ defined in (\ref{eq:W}). Thus we find
\begin{equation}\label{eq:est-comm3} \begin{split} 
\tr | [ \eps \nabla , \omega_{N,t}] | & \leq \tr \, | [ \eps \nabla , \omega_{N,0}] |\\
&\quad + \frac{1}{\eps} \int_0^t ds \, \tr | [\omega_{N,s} ,  [ V*\rho_s , \eps \nabla]] | + \frac{1}{\eps} \int_0^t ds \, \tr \, | [ \omega_{N,s}, [X_s, \eps \nabla]]|. \end{split} \end{equation}
The second term on the r.h.s. of the last equation can be controlled by
\[ \begin{split} 
\tr | [\omega_{N,s} ,  [ V*\rho_s , \eps \nabla]] | &= \eps \, \tr |[\omega_{N,s} , \nabla (V * \rho_s)]| \\ &\leq \eps \int dq \, |\widehat{V} (q)||q| |\widehat{\rho}_s (q)|  \, \tr |[ \omega_{N,s} , e^{iq \cdot x} ] | \\ &\leq \eps \left( \int dq \, |\widehat{V} (q)| (1+|q|)^{2} \right)  \sup_q \frac{1}{1+|q|} \tr  |[ \omega_{N,s} , e^{iq \cdot x} ] | \\ &\leq C \eps \, \tr |[x,\omega_{N,s}]| 
\end{split} \]
where we used the bound $\| \widehat{\rho}_s \|_\infty \leq \| \rho_s \|_1 = 1$, the estimate (\ref{eq:exp-bd}) and the assumption (\ref{eq:ass2}) on the interaction potential. As for the last term on the r.h.s. of (\ref{eq:est-comm3}), we note that, writing the exchange operator as in (\ref{eq:exX}), 
\[ \begin{split} \tr \left|[\omega_{N,s}, [ X_s, \eps\nabla]] \right| &\leq \frac{1}{N}\int dq\, |\widehat{V} (q)| \, \tr \left| [\omega_{N,s} , [ e^{iq \cdot x} \omega_{N,s} e^{-iq \cdot x}, \eps \nabla]] \right| \\ 
 &\leq \frac{2}{N} \int dq \, |\widehat{V} (q)| \tr |[ e^{iq\cdot x} \omega_{N,s} e^{-i q \cdot x}, \eps \nabla]| \\ &\leq \frac{2 \| \widehat{V} \|_1}{N} \,  \tr |[\omega_{N,s}, \eps \nabla]|. \end{split} \]
In the last inequality we used that 
\[
[ e^{iq\cdot x} \omega_{N,s} e^{-iq\cdot x}, \eps\nabla] = e^{iq \cdot x}[ \omega_{N,s} , \eps(\nabla + iq)] e^{-iq \cdot x} = e^{iq \cdot x} [\omega_{N,s}, \eps \nabla] e^{-iq\cdot x}\;.
\]
{F}rom (\ref{eq:est-comm3}), we conclude that
\[ 
\tr | [ \eps \nabla , \omega_{N,t}] | \leq \tr |[\eps \nabla , \omega_{N,0}]| +  C \int_0^t ds \, \tr  |[ x, \omega_{N,s} ] | + C N^{-2/3} \int_0^t ds \, \tr |[\eps \nabla , \omega_{N,s}]|.  \]
Summing up the last equation with (\ref{eq:deriv1}), using the conditions (\ref{eq:semi}) on the initial data and applying Gronwall's lemma, we find constants $C,c > 0$ such that 
\[ \begin{split} 
\tr | [x, \omega_{N,t}] | &\leq C N \eps \, \exp (c |t|) \quad \text{and}\\ 
\tr | [ \eps \nabla , \omega_{N,t}]  &\leq C N \eps \, \exp (c |t|). \qedhere\end{split}\]
\end{proof}

\end{document}